\documentclass[11pt,letterpaper]{article}

\usepackage[margin=1in]{geometry}
\usepackage{bbm}
\usepackage{microtype}
\usepackage{amsfonts}
\usepackage{amssymb}
\usepackage{complexity}
\usepackage{fancyhdr}
\usepackage{amsmath,enumerate}
\usepackage{wrapfig}
\usepackage{graphicx}
\usepackage{float}
\usepackage{xspace}
\usepackage{paralist}
\usepackage{enumerate}
\usepackage{cases}
\usepackage{caption}
\usepackage{algorithm}
\usepackage{listings}
\usepackage{paralist}
\usepackage{times}

\newenvironment{proof}{\noindent {\bf Proof:  }}{\hfill\rule{2mm}{2mm}}
\newenvironment{proofof}[1]{\noindent{\bf Proof of #1:  }}{\hfill\rule{2mm}{2mm}}
\newtheorem{claim}{Claim}[section]
\newtheorem{fact}{Fact}[section]
\newtheorem{definition}{Definition}[section]
\newtheorem{theorem}{Theorem}[section]
\newtheorem{lemma}{Lemma}[section]

\newcommand{\opt}{\mathsf{OPT}}
\newcommand{\clp}{\mathsf{CLP}}
\newcommand{\block}{\mathsf{blocking}}

\newcommand*\samethanks[1][\value{footnote}]{\footnotemark[#1]}

\title{On ($1$, $\epsilon$)-Restricted Max-Min Fair Allocation Problem}

\author{T-H. Hubert Chan\thanks{Department of Computer Science, the University of Hong Kong. {\texttt{\{hubert,zhtang,xwwu\}@cs.hku.hk}}} \and Zhihao Gavin Tang\samethanks
	\and Xiaowei Wu\samethanks}
\date{}

\begin{document}

\maketitle

\begin{abstract}
	We study the max-min fair allocation problem in which a set of $m$ indivisible items are to be distributed among $n$ agents such that the minimum utility among all agents is maximized.
	In the restricted setting, the utility of each item $j$ on agent $i$ is either $0$ or some non-negative weight $w_j$.
	For this setting, Asadpour et al.~\cite{talg/AsadpourFS12} showed that a certain configuration-LP can be used to estimate the optimal value within a factor of $4+\delta$, for any $\delta>0$, which was recently extended by Annamalai et al.~\cite{soda/AnnamalaiKS15} to give a
	polynomial-time $13$-approximation algorithm for the problem.
	For hardness results, Bez{\'{a}}kov{\'{a}} and Dani~\cite{sigecom/BezakovaD05} showed that it is \NP-hard to approximate the problem within any ratio smaller than $2$. 
	
	In this paper we consider the $(1,\epsilon)$-restricted max-min fair allocation problem in which each item $j$ is either heavy $(w_j = 1)$ or light $(w_j = \epsilon)$, for some parameter $\epsilon \in (0,1)$.
We show that the $(1,\epsilon)$-restricted case is also \NP-hard to approximate within any ratio smaller than $2$. 
Hence, this simple special case is still algorithmically interesting.

Using the configuration-LP, we are able to estimate the optimal value of the problem within a factor of $3+\delta$, for any $\delta>0$. Extending this idea, we also obtain a quasi-polynomial time $(3+4\epsilon)$-approximation algorithm and a polynomial time $9$-approximation algorithm. Moreover, we show that as $\epsilon$ tends to $0$, the approximation ratio of our polynomial-time algorithm approaches $3+2\sqrt{2}\approx 5.83$.
\end{abstract}

\section{Introduction}\label{sec:intro}

We consider in this paper the \emph{Max-Min Fair Allocation problem}.
A problem instance is defined by $(A,B,w)$, where $A$ is a set of $n$ agents, $B$ is a set of $m$ items and the utility of each item $j\in B$ perceived by agent $i \in A$ has weight $w_{i j}$.
An allocation of items to agents is $\sigma:B\rightarrow A$ such that $\sigma(j)=i$ iff item $j$ is assigned to agent $i$.
The max-min fair allocation problem aims at finding an allocation such that the minimum total weight received by an agent $\min_{i\in A}\sum_{j\in \sigma^{-1}(i)} w_{i j}$ is maximized.
The problem is also known as the Santa Claus Problem~\cite{stoc/BansalS06}.
In the restricted version of the problem, it is assumed that each item $j$ has a fixed weight $w_j$ such that for each $i\in A$ and $j\in B$, $w_{i j} \in \{ 0,w_j \}$, i.e.,
if an agent has non-zero utility for an item $j$, the utility is $w_j$.
We focus on this paper the restricted version of the problem (restricted allocation problem) and refer to the problem with general weights the \emph{unrestricted} allocation problem.
For the restricted allocation problem, let $B_i = \{ j\in B: w_{i j}>0 \}$ be the set of items agent $i$ is interested in.
For a collection of items $S\subseteq B$, let $w(S) = \sum_{j\in S}w_j$.

The problem can be naturally formulated as an integer program, with variable $x_{ij}$ for each $i\in A$ and $j\in B$ indicating whether item $j$ is assigned to agent $i$. Its linear program relaxation \emph{Assignment-LP} (ALP) is shown as below.
\begin{align*}
\max \qquad & T\\
\text{s.t.}\qquad  \sum_{j\in B_i} x_{ij} w_{j} &\geq T, \qquad \forall i\in A \\
 \sum_{i\in A} x_{ij} &\leq 1, \qquad \forall j\in B \\
x_{ij} &\geq 0, \qquad \forall i\in A, j\in B.
\end{align*}

Let $\opt$ be the maximum value of the restricted allocation problem such that in the optimal allocation, every agent is assigned a set of items with total weight at least $\opt$.
Bez{\'{a}}kov{\'{a}} and Dani~\cite{sigecom/BezakovaD05} showed that any feasible solution $x$ and $T$ for the ALP can be rounded into an allocation such that every agent $i$ receives at least $T - \max_{j\in B_i}w_j$ total value, which implies $\opt \geq T^* - \max_{j\in B}w_j$, where $T^*$ is the optimal value of the ALP.
However, the above result does not yield any guarantee on the integrality gap.
Actually, it can be easily shown that the integrality gap of ALP is unbounded since it is possible to have a feasible solution with $T>0$ while $\opt = 0$ (i.e., when $|B|<|A|$).
It was shown in~\cite{sigecom/BezakovaD05} that it is \NP-hard to approximate the problem within any ratio smaller than $2$ by a reduction from $3$-dimensional matching.

To overcome the limitation of ALP, a stronger linear program called \emph{Configuration-LP} (CLP) was proposed by Bansal and Sviridenko~\cite{stoc/BansalS06}, in which an $O(\frac{\log n}{\log\log n})$-approximation algorithm was obtained for the restricted allocation problem.
For any $T>0$, we call an allocation a $T$-allocation if it assigns to every agent a set of items with total weight at least~$T$. 

\begin{definition}[Bundles with Sufficient Utility]
	For all $i\in A$, the collection of bundles
	with utility at least $T$ for agent $i$ is
	$C(i,T) := \{ S\subseteq B_i: w(S)\geq T \}$.
\end{definition}

The CLP is a feasibility LP associated with $T$ indicating whether it is possible to (fractionally) assign to each agent one unit of bundle with sufficient utility. 
The LP ($\clp(T)$) and its dual are shown as follows.

\begin{minipage}[h]{0.49\textwidth}
	\centering
	\begin{align*}
	\textbf{Primal }\quad&\min \quad  0\\
	\text{s.t.}\quad \sum_{S\in C(i,T)} x_{i,S} &\geq 1, \quad \forall i\in A \\
	\sum_{i,S:j\in S\in C(i,T)} x_{i,S} &\leq 1, \quad \forall j\in B \\
	x_{i,S} &\geq 0, \quad \forall i\in A, S\in C(i,T).
	\end{align*}
\end{minipage}
\begin{minipage}[h]{0.49\textwidth}
	\begin{align*}
	\textbf{Dual }\quad\max& \quad \sum_{i\in A}y_i - \sum_{j\in B}z_j\\
	\text{s.t.}\quad y_i \leq \sum_{j\in S}z_j&, \quad \forall i\in A, S\in C(i,T) \\
	y_i \geq 0&,\quad \forall i\in A\\
	z_j \geq 0&,\quad \forall j\in B.
	\end{align*}
\end{minipage}
\vspace*{10pt}

Although $\clp(T)$ has an exponential number of variables, it is claimed in~\cite{stoc/BansalS06} that the separation problem for the dual LP is the minimum knapsack problem: given a candidate dual solution $(y, z)$, a violated constraint can be identified by finding an agent $i$ and a configuration $S\in C(i,T)$ such that $y_i > \sum_{j\in S} z_j$.
Hence, we can solve $\clp(T)$ to any desired precision.
Note that any feasible solution ${x}$ of $\clp(T)$ induces a feasible solution $\hat{x}$ for the ALP by setting $\hat{x}_{ij} = \sum_{S:j\in S\in C(i,T)} x_{i,S}\leq 1$ for all $i\in A$ and $j\in B$.

\begin{definition}[Integrality Gap]
	Let $T^*$ be the maximum value such that $\clp(T^*)$ is feasible. The ratio $\frac{T^*}{\opt}$ is known as the \emph{integrality gap}.
\end{definition}

Note that any upper bound $c$ for the integrality gap implies that we can estimate the optimal value of the problem within a factor of $c+\delta$, for any $\delta > 0$.
It is shown in~\cite{stoc/BansalS06} that the integrality gap of CLP for the unrestricted allocation problem is bounded by $O(\sqrt{n})$.
By repeatedly using the Lovasz Local Lemma, Uriel Feige~\cite{soda/Feige08} proved that the integrality gap of CLP for the restricted allocation problem is bounded by a constant.
The result was later turned into a constructive proof by Haeupler~\cite{focs/HaeuplerSS10}, who obtained the first constant approximation algorithm for the restricted allocation problem, although the constant is unspecified.
The integrality gap of CLP was later shown in~\cite{talg/AsadpourFS12} to be no larger than $4$ by a local search technique developed from Haxell~\cite{gc/Haxell95a} for finding perfect matchings in bipartite hypergraphs.
However, the algorithm is not guaranteed to terminate in polynomial time. 
%and hence do not imply a polynomial time $4$-approximation algorithm.
It is later shown by Polacek and Svensson~\cite{icalp/PolacekS12} that a simple modification of the local search algorithm can improve the running time from $2^{O(n)}$ to $n^{O(\log n)}$, which implies a quasi-polynomial $(4+\delta)$-approximation algorithm, for any $\delta>0$.
Very recently, Annamalai et al.~\cite{soda/AnnamalaiKS15} further extended the local search technique developed in~\cite{talg/AsadpourFS12, icalp/PolacekS12} for the restricted allocation problem and obtained a polynomial-time $13$-approximation algorithm for the problem.

\subsection{The $(1,\epsilon)$-Restricted Allocation Problem}

We consider in this paper the $(1,\epsilon)$-restricted allocation problem, in which each item $j\in B$ is either \emph{heavy} ($w_j = 1$) or \emph{light} ($w_j = \epsilon$, for some $\epsilon\in(0,1)$). As the simplest case of the allocation problem, the problem is not well understood.
The current best approximation results for the problem are for the restricted allocation problem.
Indeed, we believe that a better understanding of the $(1,\epsilon)$-restricted setting will shed light on improving the restricted (and even the unrestricted) allocation problem.

The $(1,\epsilon)$-restricted setting has been studied under different names.
Golovin~\cite{golovin2005max} studied the ``Big Goods/Small Goods'' max-min allocation problem, which is exactly the same as the problem we consider in this paper,  in which a small item has weight either $0$ or $1$ for each agent; a big item has weight either $0$ or $x>1$ for each agent.
They gave an $O(\sqrt{n})$-approximation algorithm for this problem and proved that it is \NP-hard to approximate the ``Big Goods/Small Goods'' max-min allocation problem within any ratio smaller than $2$ by giving a hard instance with $x=2$. We show in this paper that the inapproximability result holds for any fixed $x\geq 2$ by generalizing the hardness instance shown in~\cite{sigecom/BezakovaD05}.
Later Khot and Ponnuswami~\cite{approx/KhotP07} generalized the ``Big Goods/Small Goods'' setting and considered the $(0, 1, U)$-max-min allocation problem with sub-additive utility function in which the weight of an item to an agent is either $0$, $1$ or $U$ for some $U > 1$ and obtained an $\frac{n}{\alpha}$-approximation algorithm with $m^{O(1)}n^{O(\alpha)}$ running time, for any $\alpha\leq \frac{n}{2}$. Note that in their setting an item can have weight $1$ for an agent and $U$ for another.
In the seminal paper, Bansal and Sviridenko~\cite{stoc/BansalS06} obtained an $O(\frac{\log n}{\log\log n})$-approximation algorithm for the restricted allocation problem by first reducing the problem to the $(1,\epsilon)$-restricted case for an arbitrarily small $\epsilon>0$ while losing a constant factor on the approximation ratio, and then proving an $O(\frac{\log n}{\log\log n})$-approximation algorithm for the $(1,\epsilon)$-restricted case.

The max-min fair allocation problem is closely related to the problem of scheduling jobs on \emph{unrelated machines} to minimize \emph{makespan}, which we call the \emph{min-max allocation problem}.
The problem has the same input as the max-min fair allocation problem but aims at finding an allocation that minimizes $\max_{i\in A}\sum_{j\in \sigma^{-1}(i)} w_{i j}$.
Lenstra et al.~\cite{mp/LenstraST90} showed a $2$-approximation algorithm for the min-max allocation problem by rounding the ALP for the problem.
Applying the techniques developed for the max-min fair allocation problem, Svensson~\cite{stoc/Svensson11} gave a $\frac{5}{3}+\epsilon$ upper bound for the CLP's integrality gap of the $(1,\epsilon)$-restricted min-max allocation problem and then extended it to a $1.9412$ upper bound for the general case.
However, their algorithm is not known to converge in polynomial time.
Recently Chakrabarty et al.~\cite{soda/ChakrabartyKL15} obtained the first $(2-\delta)$-approximation algorithm for the $(1,\epsilon)$-restricted min-max allocation problem, for some constant $\delta>0$.
They considered the case when $\epsilon$ is close to $0$ since it is easy to obtain a $(2-\epsilon)$-approximation algorithm for the $(1,\epsilon)$-restricted min-max allocation problem.

Since the $(1,\epsilon)$-restriction is considered in the community to be interesting for the min-max setting, in this paper we consider this restriction for the max-min setting.

\subsection{Summary of Our Results}

We first show that we can slightly improve the hardness result of Golovin~\cite{golovin2005max} for the $(1,\epsilon)$-restricted allocation problem.
Note that in the unweighted case ($\epsilon = 1$), the problem can be solved in polynomial time by combining the max-flow computation between $A$ and $B$, with a binary search on the optimal value.
The above algorithm for the unweighted case actually provides a trivial $\frac{1}{\epsilon}$-approximation algorithm for the $(1,\epsilon)$-restricted allocation problem.
Hence, we have a polynomial-time algorithm with ratio $\frac{1}{\epsilon}<2$ for the problem when $\epsilon > 0.5$.

\begin{theorem}[Inapproximability]\label{th:hardness}
	For any $\epsilon\leq 0.5$, it is \NP-hard to approximate the $(1,\epsilon)$-restricted allocation problem within any ratio smaller than $2$.
\end{theorem}

Our reduction shows that it is \NP-hard to estimate the optimal value of the problem within any ratio smaller than $2$.
Hence, the above hardness result implies that the integrality gap of $\clp(T)$ is at least~$2$ unless $\P  = \NP$.
Actually, we are able to remove the $\P \neq \NP$ assumption by giving an instance with integrality gap $2$ in Section~\ref{sec:hardness}.

For the restricted allocation problem, the best hardness result on the approximation ratio is $2$ while the best upper bound for the integrality gap of $\clp(T)$ is $4$.
It is not known which bound (or none) is tight.
As a step towards closing this gap, we analyze the integrality gap of $\clp(T)$ for the $(1,\epsilon)$-restricted case and show that the upper bound of $4$ is not tight in this case (Section~\ref{sec:Integrality-gap}).
Our result on the integrality gap upper bound implies that in polynomial time we can estimate $\opt$ for the $(1,\epsilon)$-restricted allocation problem within a factor of $3+\delta$, for any $\delta > 0$.

\begin{theorem}[Integrality Gap]\label{th:integrality-gap}
	The integrality gap of the configuration-LP of the $(1,\epsilon)$-restricted allocation problem is at most $3$.
\end{theorem}

We also observe that by picking the ``closest addable edge'', the running time of the local search algorithm can be improved to quasi-polynomial (Section~\ref{ssec:quasi-approx}). The idea was first used by Polacek and Svensson~\cite{icalp/PolacekS12} to obtain the $(4+\delta)$-approximation algorithm for the restricted allocation problem. However, instead of constructing feasible dual solutions for $\clp(T)$, our analysis is based on the assumption of $T\leq \opt$ and is a direct extension of our proof on the integrality gap of $\clp(T)$.

\begin{theorem}[Quasi-Polynomial-Time Approximation]\label{th:quasi-3}
	There exists a $(3+4\epsilon)$-approximation algorithm for the $(1,\epsilon)$-restricted allocation problem that runs in $n^{O(\frac{1}{\epsilon}\log n)}$ time.
\end{theorem}

We further extend the quasi-polynomial approximation algorithm by combining the lazy update idea of~\cite{soda/AnnamalaiKS15} to obtain a polynomial approximation algorithm (Section~\ref{sec:poly-approx}).

\begin{theorem}[Polynomial-Time Approximation]\label{th:poly-9}
	For any $\epsilon\in(0,1)$, there exists a polynomial-time $9$-approximation algorithm for the $(1,\epsilon)$-restricted allocation problem. Moreover, the approximation ratio approaches $3+2\sqrt{2}\approx 5.83$ as $\epsilon$ tends to $0$.
\end{theorem}

Interestingly, while our quasi-polynomial- and polynomial-time algorithms are extended from the integrality gap analysis by combining ideas on improving the running time of local search, unlike existing techniques, our algorithms and analysis do not directly use the feasibility of $\clp(T)$.
To lead to contradictions, existing results~\cite{icalp/PolacekS12, soda/AnnamalaiKS15} tried to construct feasible dual solutions for $\clp(T)$ with positive objective values (which implies the infeasibility of $\clp(T)$).
In contrast, our analysis shows that as long as $T\leq \opt$, our algorithms terminate with the claimed approximation ratios, which simplifies the analysis and is an advantage in some cases when $\clp(T)$ cannot be applied, i.e., when the utility function is sub-additive~\cite{approx/KhotP07}.

\subsection{Other Related Work}

\paragraph{Unrestricted Allocation Problem.}
Based on Bansal and Sviridenko's proof~\cite{stoc/BansalS06} of $O(\sqrt{n})$-integrality gap for the unrestricted allocation problem, Asadpour and Saberi~\cite{stoc/AsadpourS07} achieved an $\tilde{O}(\sqrt{n})$-approximation algorithm.
The current best approximation result for the problem is an $\tilde{O}(n^{\delta})$-approximation algorithm that runs in time $n^{O(\frac{1}{\delta})}$, for any $\delta=\Omega(\frac{\log\log n}{\log n})$, obtained by Chakrabarty et al.~\cite{focs/ChakrabartyCK09}.

\paragraph{Other Utility Functions.}
The max-min fair allocation problem with different utility functions has also been considered. 
Golovin~\cite{golovin2005max} gave an $(m-n+1)$-approximation algorithm for the problem when the utility functions of agents are submodular.
We note that their result can also be extended to sub-additive utility functions.
Khot and Ponnuswami~\cite{approx/KhotP07} also considered the problem with sub-additive utility functions and obtained a $(2n-1)$-approximation algorithm.
Later Goemans and Harvey~\cite{soda/GoemansHIM09} obtained an $\tilde{O}(n^{\frac{1}{2}+\delta})$-approximation for submodular max-min allocation problem in $n^{O(\frac{1}{\delta})}$ time using the $\tilde{O}(n^{\delta})$-approximation algorithm by Chakrabarty et al.~\cite{focs/ChakrabartyCK09} as a black box.

\section{Integrality Gap for Configuration LP}\label{sec:Integrality-gap}

We show in this section that
for the $(1,\epsilon)$-restricted allocation problem, the integrality gap of the $\clp$ is at most $3$.
Fix $T > 0$ be such that $\clp(T)$ is feasible.

We show that whenever $\clp(T)$ is feasible, there exists a $\frac{T}{3}$-allocation (hence $\opt\geq \frac{T}{3}$), which implies an integrality gap of at most $3$.
Given any solution $x$ for $\clp(T)$ and the induced ALP solution $\hat{x}$,
for all $\hat{x}_{ij}=0$, we can remove $j$ from $B_i$ (pretending that $i$ is not interested in $j$). This operation will preserve the feasibility of ${x}$ while (possibly) decreasing $\opt$, which could only enlarge the integrality gap.
From now on we assume that a positive fraction of every item in $B_i$ is assigned to agent $i$.

\paragraph{Assumption on $T$.}
To achieve a $\frac{T}{3}$-allocation, we can assume that $T<\frac{3}{2}$; otherwise, we can  get a $T-1\geq \frac{T}{3}$ allocation by rounding the ALP solution $\hat{x}$~\cite{sigecom/BezakovaD05}.
We can further assume $T\geq 1$ since otherwise we can set all weights $w_j\geq T$ to $T$ (which does not change $\clp(T)$) and scale all weights so that the maximum weight is $1$.
From now on, we assume that $T \in [ 1,\frac{3}{2} )$ and $\clp(T)$ is feasible.

Let $k= \lceil \frac{T}{\epsilon} \rceil$. Note that every
bundle consisting solely of light items must contain at least $k$ items to have sufficient utility.
For all $i\in A$, let $B^1_i = \{ j\in B_i: w_j = 1\}$ be the set of {heavy} items and $B^\epsilon_i = \{ j\in B_i: w_j = \epsilon \}$ be the set of {light} items.
Our algorithm fixes an integer $r=\lceil\frac{k}{3}\rceil$ and tries to assign items such that each agent $i$ receives either a heavy item $j\in B^1_i$ or $r$ light items in $B^\epsilon_i$. Suppose we are able to find such an allocation, then the integrality gap is $\frac{T}{r\epsilon}\leq 3$.

\subsection{Getting a ``Minimal'' Solution}

Let ${x}^*$ be a solution for $\clp(T)$.
We create another solution ${x}$ (which might not be feasible) as follows.
Initialize $x_{i,S} = 0$ for all $i\in A$ and $S\subseteq B_i$.
For all $x^*_{i,S} > 0$, where $S\in C(i,T)$,
\begin{enumerate}
	\item if $S' = S\cap B^1_i \neq \emptyset$, then set $ x_{i,S'} = x^*_{i,S}$;
	\item otherwise, $S$ contains only light items and set $x_{i, S} = x^*_{i,S}$.
\end{enumerate} 

Note that for each $i\in A$ we have the following properties on ${x}$:
\begin{enumerate}
	\item (heavy/light configurations) if $x_{i,S}>0$, then ($S\subseteq B^1_i \wedge|S|\geq 1$) or ($S\subseteq B^\epsilon_i \wedge |S|\geq k$).
	\item (covering constraint for agent) $\sum_{S\subseteq B_i} x_{i,S} =  \sum_{S\in C(i,T)} x^*_{i,S} \geq 1$.
	\item (packing constraint for item) for all $j\in B$: $\sum_{i,S:j\in S} x_{i,S} \leq \sum_{i,S:j\in S\in C(i,T)} x^*_{i,S} \leq 1$.
\end{enumerate}

Now we construct a hypergraph $H(A\cup B, E)$ as follows: for all $x_{i,S}> 0$,
\begin{enumerate}
	\item if $S\subseteq B^1_i$, then for each $j\in S$, add $\{i,j\}$ to $E$ (we call such an edge {\bf heavy});
	\item otherwise for each $S'\subseteq S$ and $|S'|= r$, add $\{i\} \cup S'$ to $E$ (we call such an edge {\bf light}).
\end{enumerate}

A matching $M\subseteq E$ is a collection of disjoint edges. 
Note that any perfect matching of $H$ that matches all nodes in $A$ provides an allocation that assigns each $i\in A$
either a heavy item or $r$ light items.
For all $F\subseteq E$, let $A(F) = A\cap(\bigcup_{e\in F}e)$ and $B(F) = B\cap(\bigcup_{e\in F}e)$.

\subsection{Finding a Perfect Matching}

Recall that the existence of a perfect matching
in $H(A \cup B, E)$ such that every agent in $A$ is matched
implies that the integrality gap of $\clp(T)$ is at most 3.

\begin{theorem}
\label{th:perfect}
	The above hypergraph $H(A\cup B, E)$ has a perfect matching.
\end{theorem}

\begin{proof}
	Given a partial matching $M\subseteq E$, we show how to extend its cardinality by one if $|M|\leq |A|-1$.	
	Let $i_0\in A\backslash A(M)$ be an agent not matched by $M$.
	For the initial step, suppose $X_1$ contains an arbitrary edge $e_1$ with $A(e_1) = \{ i_0 \}$
	and $Y_1 = \block(e_1) = \{ f\in M: B\cap e_1\cap f \neq \emptyset\}$ be the blocking edges of $e_1$.
	If $\block(e_1) = \emptyset$, then we can add edge $e_1$ to the matching. Assume $\block(e_1) \neq \emptyset$.
	
For the recursive step,	suppose we already have edges $X_t$ (where $t = |X_t|$) and $Y_t$, which together form an alternating tree rooted at $i_0$.
%Suppose we have added $t$ edges in the tree. 
We consider adding the $(t+1)$-st edge to $X_t$ as follows.
	An edge $e\in E$ is {\bf addable} if (1) $A(e)\in A(X_{t}\cup Y_{t})$; (2) $B(e)\cap B(X_{t}\cup Y_{t}) = \emptyset$. If such an edge $e_{t+1}$ exists and $\block(e_{t+1})\neq \emptyset$, let $X_{t+1} = X_t\cup\{ e_{t+1} \}$ and $Y_{t+1} = Y_t\cup \block(e_{t+1})$.
	If $\block(e_{t+1}) = \emptyset$, then we {\bf contract} $X_t$ by swapping out blocking edges (the details of contraction will be discussed later).
	The contraction operation guarantees that every addable edge has at least one blocking edge.
	
	\begin{claim}[Always Addable]
		There is always an addable edge $e_{t+1}$.
	\end{claim}
	\begin{proof}
		Let $P=A(X_t\cup Y_t)$ be the agents in the tree.
		Note that $|P| = |Y_t|+1$ since each agent $i\neq i_0$ in $P$ has an unique blocking edge that {\bf introduces} $i$.
		
		Let $X^1_t$ ($Y^1_t$) be the heavy edges and $X^\epsilon_t$ ($Y^\epsilon_t$) be the light edges of $X_t$ ($Y_t$).
		
		We have $|X^1_t| = |Y^1_t|$ since heavy edges can only be blocked by heavy edges. 
		We have $|X^{\epsilon}_t| \leq |Y^{\epsilon}_t|$ since each addable edge has at least one blocking edge.

		Let $x^1_P = \sum_{i\in P}\sum_{S\subseteq B^1_i} x_{i,S}$ be the total units of heavy bundles assigned to $P$ by $x$, which is a lower bound for the total number of heavy items $B^1_P = \cup_{i\in P}B^1_i$ agents in $P$ are interested in since
		\begin{equation*}
		x^1_P = \sum_{i\in P}\sum_{S\subseteq B^1_i} x_{i,S} \leq
		\sum_{i\in P}\sum_{S\subseteq B^1_i} \sum_{j\in S} x_{i,S}
		= \sum_{j\in B^1_P} \sum_{i,S: j\in S\subseteq B^1_i} x_{i,S} \leq |B^1_P|.
		\end{equation*}
		
		Let $x^\epsilon_P = \sum_{i\in P}\sum_{S\subseteq B^\epsilon_i} x_{i,S}$ be the total units of light bundles assigned to $P$ by $x$.
		By construction of $x$, we have
		\begin{align*}
		\sum_{i\in P}\sum_{S\subseteq B_i} x_{i,S} = \sum_{i\in P}\sum_{S\subseteq B^1_i} x_{i,S} + \sum_{i\in P}\sum_{S\subseteq B^\epsilon_i} x_{i,S} = x^1_P + x^\epsilon_P  \geq |P|.
		\end{align*}
		
		Since $|Y^1_t|$ heavy items are already introduced in the tree, if $x^1_P > |Y^1_t|$, then there must exist an addable heavy edge for some $i\in P$. If $x^1_P \leq |Y^1_t|$, then we have $x^\epsilon_P \geq |P| - x^1_P \geq |Y^\epsilon_t| +1 \geq |X^\epsilon_t|+1$.
		Note that every light addable edge has at most $r-1$ unblocked items, the total number of light items in the tree is
		\begin{equation}
		\label{eq:light}
		|B^\epsilon(X_t\cup Y_t)| \leq (r-1)|X^\epsilon_t| + r|Y^\epsilon_t|
		\leq (2r-1)(x^\epsilon_P - 1) < (2r-1)x^\epsilon_P.
		\end{equation}
		
		For each $i\in P$ and $S\subseteq B^\epsilon_i$, if $x_{i,S}>0$, then by construction we have $|S| \geq k \geq 3r -2$. If $i$ has no more addable light edges (has at most $r-1$ {\bf unintroduced} light items in $H$), then at least
		$$
		\sum_{S\subseteq B^\epsilon_i} (|S|-(r-1))x_{i,S} \geq (2r-1)\sum_{S\subseteq B^\epsilon_i} x_{i,S}
		$$
		units of configurations of light items appear in the tree.
		If there is no more addable light edges for all $i\in P$, then we have
		\begin{equation*}
		|B^\epsilon(X_t\cup Y_t)| \geq \sum_{j\in B^\epsilon(X_t\cup Y_t)}\sum_{i,S:j\in S\subseteq B^\epsilon_i} x_{i,S} \geq  (2r-1)\sum_{i\in P} \sum_{S\subseteq B^\epsilon_i} x_{i,S} = (2r-1)x^\epsilon_P,
		\end{equation*}
		which is a contradiction to (\ref{eq:light}).
	\end{proof}
	
\paragraph{Contraction.}
	If $\block(e_{t+1}) = \emptyset$, then we remove the blocking edge $f$ that introduces $A(e_{t+1})$ from the matching and include $e_{t+1}$ into the matching. Both $e_{t+1}$ and $f$ are removed from the tree.
	We also remove all edges added after $f$ since they can possibly be introduced by $A(f)$.
	We call this operation a contraction on $e_{t+1}$.
	Note that this operation reduces the size of $\block(e')$ by one, for the edge $e'$ that is blocked by $f$.
	If $\block(e')=\emptyset$ after that, then we further contract $e'$ recursively.
	After all contractions, suppose the remaining addable edges in the tree are $e_1,e_2,\ldots,e_{t'}$ (ordered by the time they are added to the tree),
	we set $t=t'$, $X_{t'}$ and $Y_{t'}$ be the addable and blocking edges, respectively.
	
\paragraph{Signature.} At any moment before including an addable edge (suppose there are $t$ addable edges in the tree), let $s_i = |\block(e_i)|$ for all $i\in[t]$.
	Let $\textbf{s} = (s_1,s_2,\ldots,s_t,\infty)$ be the signature of the tree. Then, we have the following.
	\begin{enumerate}
		\item The lexicographical value of $\textbf{s}$ reduces after each iteration. If there is no contraction in the iteration, then in the signature,
		the $(t+1)$-st coordinate decreases from $\infty$ to $s_{t+1}$,
while $s_i$ remains the same for all $i\leq t$. 
Otherwise, let $e_i$ be the edge whose number of blocking edges is reduced by one but remains positive in the contraction phase. Then, we have $s_i$ is reduced by one while $s_j$ remains the same for all $j<i$.
		
\item There are at most $2^n$ different signatures since $\sum_{i\in[t]} s_i \leq n$ and $t\leq n$. 
	\end{enumerate}
	Since an addable edge can be found in polynomial time and the contraction operation stops in polynomial time, a perfect matching can be found in $n\cdot 2^n\cdot \poly(n)$ time.
\end{proof}

\section{Quasi-Polynomial-Time Approximation Algorithm}\label{ssec:quasi-approx}

We show in this section that a simple modification on the algorithm for finding a perfect matching in Section~\ref{sec:Integrality-gap} can dramatically improve the running time from $2^{O(n)}$ to $n^{O(\log n)}$.
Assume that $T\leq \opt$.
Note that in this case we can still assume $T\in[ 1,\frac{3}{2} )$.

Note that combining the polynomial time $\frac{1}{\epsilon}$-approximation algorithm, the approximation ratio we obtain in quasi-polynomial time is $\min\{ \frac{1}{\epsilon}, 3+4\epsilon \}\leq 4$ for all $\epsilon\in(0,1)$. Moreover, when $\epsilon\rightarrow 0$ (in which case the problem is still $(2-\delta)$-inapproximable), our approximation ratio approaches the integrality gap upper bound $3$.

\begin{proofof}{Theorem~\ref{th:quasi-3}}
Let $T$ be a guess of $\opt$ and $k = \lceil \frac{T}{\epsilon} \rceil$.
Since the statement trivially holds for $\epsilon\geq \frac{1}{4}$ ($\frac{1}{\epsilon}\leq 3+4\epsilon$). We assume that $\epsilon<\frac{1}{4}$ (which means $k\geq 5$).
We show that if $T\leq \opt$, then we can find in quasi-polynomial time a $\frac{T}{3+4\epsilon}$-allocation;
if no such allocation is found after the time limit,
then $T$ should be decreased as in binary search.
Let $r = \lceil \frac{k}{3+4\epsilon} \rceil$.
To prove the theorem, it suffices to show that a feasible allocation that assigns to each agent $i$ either a heavy item in $B^1_i$ or $r$ light items in $B^\epsilon_i$ can be found in $n^{O(\log n)}$ time, for any $\epsilon<\frac{1}{4}$. 
%Since our goal is to assign to each agent $i$ either a heavy item or $r$ light items, 
We define a heavy edge $\{i,j\}$ for each $j\in B^1_i$ and
a light edge $\{i\} \cup S$ for each $S\subseteq  B^\epsilon_i$ and $|S|=r$.
% (instead of defining solution $x$ and hypergraph $H$).

As in the proof of Theorem~\ref{th:perfect},
we wish to find a perfect matching for all agents in $A$.
Suppose in some partial matching, there is an unmatched
agent $i_0$ and we construct an alternating tree rooted at $i_0$.
For each addable edge $e$, we denote by $d(e)$ the number of light edges (including $e$) in the {\bf path} from $i_0$ to $e$ in an alternating tree rooted at $i_0$.
Note that a path is a sequence of edges alternating between addable edges and blocking edges.
The algorithm we use in this section is the same as previous, except that when there are addable edges, we always pick the one $e$ such that the {\bf distance} $d(e)$ is {\bf minimized}.
We show that in this case there is always an addable edge within distance $O(\frac{1}{\epsilon}\log n)$.

Let $X_i$ and $Y_i$ be the set of addable edges and blocking edges at distance $i$ from $i_0$, respectively. Note that $Y_i = \emptyset$ for all odd $i$ since light blocking edge must be introduced due to light addable edge. Moreover, since on the path from $i_0$ to every addable edge $e\in X_i$, the light edge (if any) closest to $e$ must be a blocking edge (of even distance), we know that $X_\text{odd}$ contains only light edges and $X_\text{even}$ contains only heavy edges. Let $Y^1_i$ and $Y^\epsilon_i$ be the set of heavy edges and light edges in $Y_i$, respectively.

Let $L= \lceil \log_{1+\frac{\epsilon}{10}}n \rceil$.
It suffices to prove Claim~\ref{claim:ep-increase} below since it implies that
\begin{equation*}
|Y^\epsilon_{\leq 2L+2}|>(1+\frac{\epsilon}{10})|Y^\epsilon_{\leq 2L}|>(1+\frac{\epsilon}{10})^{L}|Y^\epsilon_2|\geq n,
\end{equation*}
which is a contradiction and implies that  there is always an addable edge within distance $2L+1$.
Note the the last inequality also comes from Claim~\ref{claim:ep-increase} since otherwise $|Y^\epsilon_2|=0$ and $|Y^\epsilon_4|=0\leq \frac{\epsilon}{10} |Y^\epsilon_2|$ would be a contradiction.

\begin{claim}\label{claim:ep-increase}
	For all $l\in[L]$, when there is no more addable edge within distance $2l+1$, we have $|Y^\epsilon_{2l+2}|>\frac{\epsilon}{10} |Y^\epsilon_{\leq 2l}|$.
\end{claim}
\begin{proof}
	Let $P = A(X_{\leq 2l}\cup Y_{\leq 2l}) = A(Y_{\leq 2l})\cup\{ i_0 \}$.
	Since there is no more addable edges within distance $2l+1$, we know that every agent $i\in P$ does not admit any addable edges. Hence for each $i\in P$, all heavy items in $B^1_i$ are already included in $B^1(X^1_{\leq 2l})$ and at most $r-1$ light items in $B^\epsilon_i$ are not included in $B^\epsilon(X^\epsilon_{\leq 2l+1}\cup Y^\epsilon_{\leq 2l+2})$.
	
	Since $T\leq \opt$, we know that at least $|P|-|B^1(X^1_{\leq 2l})| = |Y^\epsilon_{\leq 2l}|+1$ agents in $P$ are assigned only light items.
	Hence, out of at least $k$ light items assigned to each of those agents, at least $k-r+1$ items must be included in $B^\epsilon(X^\epsilon_{\leq 2l+1}\cup Y^\epsilon_{\leq 2l+2})$, which means
	\begin{equation*}
	|B^\epsilon(X^\epsilon_{\leq 2l+1}\cup Y^\epsilon_{\leq 2l+2})|\geq (k-r+1)(|Y^\epsilon_{\leq 2l}|+1).
	\end{equation*}
	
	%Let $x^1_P = \sum_{i\in P}\sum_{S\subseteq B^1_i}x_{i,S}$ and $x^\epsilon_P = \sum_{i\in P}\sum_{S\subseteq B^\epsilon_i}x_{i,S}$. Then we have $\sum_{i\in P}\sum_{S\subseteq B_i}x_{i,S} = x^1_P+x^\epsilon_P \geq |P| = |Y_{\leq 2l}|+1$.
	%Since there is no more heavy addable edge within distance $2l$, we have $x^1_P \leq |Y^1_{\leq 2l}|$ and hence $x^\epsilon_P \geq |Y_{\leq 2l}|+1-|Y^1_{\leq 2l}| = |Y^\epsilon_{\leq 2l}|+1$.
	
	Assume $|Y^\epsilon_{2l+2}|\leq \frac{\epsilon}{10} |Y^\epsilon_{\leq 2l}|$, we have $|Y^\epsilon_{\leq 2l+2}|\leq (1+\frac{\epsilon}{10}) |Y^\epsilon_{\leq 2l}|$.
	Since every addable edge contains at most $r-1$ unblocked items (items not used by $M$), we have the following upper bound for the number of light items in the tree:
	\begin{equation*}
	|B^\epsilon(X^\epsilon_{\leq 2l+1}\cup Y^\epsilon_{\leq 2l+2})|\leq (r-1)|X^\epsilon_{\leq 2l+1}|+r|Y^\epsilon_{\leq 2l+2}|\leq (1+\frac{\epsilon}{10})(2r-1)|Y^\epsilon_{\leq 2l}|.
	\end{equation*}

	%For each $i\in P$, $x_{i,S}>0$ and $S\subseteq B^\epsilon_i$, we have $|S|\geq k$. Since at most $r-1$ light items in $S$ are not included in $|B^\epsilon(X^\epsilon_{\leq 2l+1}\cup Y^\epsilon_{\leq 2l+2})|$ (otherwise there exists an addable light edge within distance $2l+1$), we have the following lower bound:
	
	For $\epsilon<\frac{1}{4}$, $T\in[ 1,\frac{3}{2} )$, $k =\lceil \frac{T}{\epsilon} \rceil$ and $r = \lceil \frac{k}{3+4\epsilon} \rceil$, we have $k\geq 3\lceil \frac{k}{3} \rceil - 2\geq 3r-2$. Suppose $k = 3r-2$, then we have
	\begin{equation*}
	k = 3\lceil \frac{k}{3+4\epsilon} \rceil-2\leq \frac{3k}{3+4\epsilon}+1 = k-(\frac{4\epsilon k}{3+4\epsilon}-1),
	\end{equation*}
	which is a contradiction since $\frac{4\epsilon k}{3+4\epsilon}>1$. Hence we have $k\geq 3r-1$, which implies
	\begin{equation*}
	k-r+1 \geq (3r-1)-(r-1) = 2r \geq  (1+\frac{\epsilon}{10})(2r-1)
	\end{equation*}
	since $r\leq \frac{5}{\epsilon}$. Hence we have a contradiction.
\end{proof}

At any moment before adding an addable edge, suppose we have constructed $X_{\leq 2l}$ and $Y_{\leq 2l}$.
By the above argument we have $2l\leq 2L = O(\frac{1}{\epsilon}\log n)$.
Let $a_i = -|X_i|$.
Let $|Y^1_i|=b_i$ and $|Y^\epsilon_i| = b_{i-1}$ for all even $i$.
Let $\textbf{s}  = (a_0,b_0,a_1,b_1,\ldots,a_{2l},b_{2l},\infty)$ be the {\bf signature} of the alternating tree.
We show that $\textbf{s}$ is lexicographically decreasing accross all iterations.

\paragraph{No contraction.} Suppose we added an addable edge $e$ with $\block(e)\neq \emptyset$, then $e$ will be included in $X_{\leq 2l}$ or a newly constructed $X_{2l+1}$, in both cases the lexicographic value of $\textbf{s}$ decreases since the last modified coordinate decreases.

\paragraph{Contraction.} Suppose the newly added edge has no blocking edge, then in the contraction, let $f\in Y^\epsilon_{2i}$, which must be light, be the last blocking edge that is removed. Since $b_{2i-1}$ decreases while $a_{j}$ (for all $j\leq 2i-1$) and $b_{j}$ (for all $j\leq 2i-2$) do not change, the lexicographic value of $\textbf{s}$ decreases.

Since $L= O(\frac{1}{\epsilon}\log n)$, there are $n^{O(\frac{1}{\epsilon}\log n)}$ different signatures.
Since an addable edge can be found in polynomial time and the contraction operation stops in polynomial time, the running time of the algorithm is $n\cdot \poly(n)\cdot n^{O(\frac{1}{\epsilon}\log n)} = n^{O(\frac{1}{\epsilon}\log n)}$.
\end{proofof}

\section{Polynomial-Time Approximation Algorithm}\label{sec:poly-approx}

We give a polynomial-time approximation algorithm in this section.
Based on the previous analysis, to improve the running time from $n^{O(\log n)}$ to $n^{O(1)}$, we need to bound the total number of iterations (signatures) by $\poly(n)$.
On a high level, our algorithm is similar to that of Annamalai et al.~\cite{soda/AnnamalaiKS15}: we apply the idea of lazy update and greedy player such that after each iteration, either a new layer is constructed or the size of the highest layer changed is reduced by a constant factor.
However, instead of constructing feasible dual solutions, we extend the charging argument used in the previous sections on counting the number of light items in the tree to prove the exponential growth property of the alternating tree.
Moreover, by avoiding the use of $\clp(T)$ (and its dual), we are able to provide a simpler analysis of the algorithm while achieving a better approximation ratio.

In binary search, let $T$ be a guess of $\opt$. As explained earlier,
we can assume $T\in [ 1,\frac{3}{2} )$. Let $k = \lceil \frac{T}{\epsilon} \rceil$.
%Note that any configuration without heavy item contains at least $k$ light items.
Our algorithm aims at assigning to each agent either a heavy item or $r$ light items, for some fixed $r \leq \frac{k}{2}$ when $T\leq \opt$.
Such an allocation gives a $\frac{k}{r}$-approximation.
%If $\frac{k}{r}\leq \rho$, then we have a $\rho$-approximation.
Let $p \in(r,k)$ be an integer parameter.
Let $0<\mu\ll1$ be a very small constant, i.e., $\mu = 10^{-10}$.

As before, for each $i\in A$, we call $\{i,j\}$ a heavy edge for  $j\in B^1_i$, and
 $\{i\} \cup S$ a light edge if $S\subseteq  B^\epsilon_i$.
 However, in this section,
we use two types of light edges: either $|S| = p$ (addable edges) or $|S| = r$ (blocking edges).
%Let $H(A\cup B, E)$ be a hypergraph.
%For all $i\in A$, we call $\{ i,j \}\in E$ a {\bf heavy edge} for all $j\in B^1_i$ and $\{ i \} \cup R\in E$ a {\bf light edge} for all $R\subseteq B^\epsilon_i$ with $|R| = r$.
%We show that we can find in polynomial time a perfect matching $M$ with $A(M) = A$ for $H$, which implies a $\frac{k}{r}$ approximation allocation.
Let $M$ be a maximum matching between $A$ and $B^1$.
We can regard $M$ as a partial allocation that assigns maximum number of heavy items.
Let $i_0$ be an unmatched node in $M$.
We can further assume that every heavy item is interesting to at least $2$ agents since otherwise we can assign it to the only agent and remove the item and the agent from the problem instance.
We use $``+"$ and $``-"$ to denote the inclusion and exclusion of singletons in a set, respectively.

\subsection{Flow Network}

Let $G(A\cup B^1, E_M)$ be a {\bf directed} graph uniquely defined by $M$ as follows.
For all $i\in A$ and $j\in B^1_i$, if $\{i,j\}\in M$ then $(j,i)\in E_M$, otherwise $(i,j)\in E_M$.
We can interpret the digraph as the residual graph of the ``interest'' network (a digraph with directed edges from each $i$ to $j\in B^1_i$) with current flow $M$.
The digraph $G$ has the following properties
%(since $M$ assigns maximum number of heavy items):
\begin{compactitem}
	\item every $i\in A$ has in-degree $\leq 1$, every $j\in B^1$ has out-degree $\leq 1$ and in-degree $\geq 1$.
	\item all heavy items reachable from $i\in A$ with in-degree $0$ must have out-degree $1$ (otherwise we can augment the size of $M$ by one).
\end{compactitem}

Given two sets of light edges $Y$ and $X$ ($A(Y)$ and $A(X)$ do not have to be disjoint), let $f(Y,X)$ denote the maximum number of node-disjoint paths in $G(A\cup B^1, E_M)$ from $A(Y)$ to $A(X)$.
Let $F(Y,X)$ be those paths.
We will later see that each such path alternates between heavy edges and their blocking edges.
Unlike the quasi-polynomial-time algorithm, in our polynomial-time algorithm, the heavy edges do not appear in the alternating tree. Instead, they are used in the flow network $G(A\cup B^1, E_M)$ to play a role of connecting existing addable light edges and blocking light edges.

\subsection{Building Phase}

\begin{definition}[Layers]
For all $i\geq 1$, a layer $L_i$ is a tuple $(X_i, Y_i)$, where $X_i$ is a set of addable edges and $Y_i$ is a set of blocking edges that block edges in $X_i$.
\end{definition}

Initialize $l= 0$, $L_0 = ( \emptyset, \{(i_0,\emptyset)\} )$.
%We call an addable edge $e$ {\bf unblocked} if $\block(e) = \emptyset$.
We call an addable edge $e = \{i\} \cup P$ {\bf unblocked} if it contains at least $r$ unblocked light items: $|P\backslash(\bigcup_{e'\in\block(e)} B^\epsilon(e'))|\geq r$.
Initialize the set of unblocked addable edges be $I=\emptyset$.
Throughout the whole algorithm, we maintain a set $I$ of unblocked addable edges and layers $L_i(X_i, Y_i)$ for all $i\leq l$, where $X_i$ contains blocked addable edges.
Initialize $X_{l+1} = Y_{l+1} = \emptyset$. We build a new layer as follows.

\begin{definition}[Addable]
Given layers $X_{\leq l+1}$ and $Y_{\leq l}$, an edge $e= \{i\} \cup P$ is addable if $|P|=p$ and $P\subseteq B^\epsilon_i\backslash B^\epsilon(X_{\leq l+1}\cup Y_{\leq l})$ such that
%\begin{equation*}\label{eq:include-I}
$f(Y_{\leq l}, X_{\leq l+1}\cup I+ e) > f(Y_{\leq l}, X_{\leq l+1}\cup I)$.
%\end{equation*} 
\end{definition}

Note that such an edge is connected to a blocking edge in $Y_{\leq l}$ by a path in $G(A\cup B^1, E_M)$ that is disjoint from other paths connecting existing blocking edges and addable edges.

Given an addable edge: if it is unblocked, then include it in $I$; otherwise include it in $X_{l+1}$.
When there is no more addable edges, let $Y_{l+1} = \block(X_{l+1}) = \bigcup_{e\in X_{l+1}} \block(e)$, set $l=l+1$ and try to {\bf contract} $L_{l}$.
Note that it is possible that a blocking edge blocks multiple addable edges.

%{\color{blue} Blocking edges are added after all addable are introduced. Hence no superposed-build. If we add blocking edges immediately after each addable edge then each blocking edge has only one addable edge that it blocks, which helps us to replace $\frac{p-r+1}{r}$ with $\lceil \frac{p-r+1}{r} \rceil$.
%In contraction of $Y_t$, we need to super-pose build $X_t$, which can possibly increase the size of $Y_t$. Hence in this case we need to guarantee that the increase of $X_t$ in super-pose build is at least $1+\mu^2$ and the signature need to include both $X_t$ and $Y_t$, i.e., $(\ldots, -\log |X_t|, \log |Y_t| ,\ldots)$ (refer to SODA'16).}

\subsection{Collapse Phase}

Let $W = F(Y_{\leq l}, I)$ be constructed as follows.
Initialize $W = \emptyset=F(Y_{\leq 0}, I)$. Recursively for $i = 1,2,\ldots,l$, let $W = F(Y_{\leq i}, I)$ be augmented from $W=F(Y_{\leq i-1}, I)$.
In the final $W$, let $W_i\subseteq W$ be the paths from $A(Y_i)$ to $A(I)$ and let $I_i\subseteq I$ be those reached by $W_i$.
By the above construction, if $f\in Y_{\leq i}$ have no out-flow in $F(Y_{\leq i}, I)$, then it will not have out-flow in $F(Y_{\leq j}, I)$, for any $j>i$.
Hence we have for all $i = 1,2,\ldots,l$, $|W_i| = |I_i|$ and $|W_{\leq i}| = |I_{\leq i}| = f(Y_{\leq i}, I)= f(Y_{\leq i}, I_{\leq i})$.

Note that every path in $W_i$ starts with an agent $u\in A(Y_i)$ that is assigned a light edge by $M$ and ends at a agent $v\in A(I_i)$ with an unblocked addable edge, which provides a possibility of swapping out a blocking edge in the tree with an unblocked addable edge (by reassigning all heavy items in the path).

\begin{definition}[Collapsible]
	We call layer $L_i$ collapsible if $|I_i|\geq \mu |Y_i|$.
\end{definition}

Intuitively, $|I_i|\geq \mu |Y_i|$ implies that we can swap out a $\mu$ fraction of blocking edges in $Y_i$ (which is called a collapse).
Let $L_t$ be the earliest collapsible layer, we collapse it as follows.

\noindent
{\bf Step-(1).}
%Let $W' = F(Y_{\leq t-1}, X_t\cup I_{\leq t-1})$ be augmented from $W_{\leq t-1}$. Note that we have
%\begin{compactitem}
%	\item all $I_{\leq t-1}$ are reached by $W'$.
%	\item $|W'| = f(Y_{\leq t-1}, X_t\cup I)$ (all $I_{\geq t}$ are not reachable after $W_{\leq t-1}$).
%	\item $W'$ and $W_t$ are node-disjoint (all nodes in $W_t$ are not reachable, otherwise $I_t$ reachable).
%\end{compactitem}
%
For each path $P(u,v)$ in $W_t$ from $e_1 :=\{u\} \cup R \in Y_t$ to $e_2 := \{v\} \cup P \in I_t$:
\begin{enumerate}
	\item $M = M - e_1 + e'$, swap out blocking edge $e_1$ with $e' := \{v\} \cup P'$, where $|P'|=r$ and $P'\subseteq P\backslash \bigcup_{e\in\block(e_2)}B^\epsilon(e)$,
	\item reverse all heavy edges in $P(u,v)$: $M = M \cup \{ \{i,j\}: (i,j)\in P(u,v)\cap(A\times B) \}\backslash \{ \{i',j'\}: (j',i')\in P(u,v)\cap(B\times A) \} $.
\end{enumerate}

Note that after the above operations, only $Y_t$ and $M$ are changed: the size $|Y_t|$ is reduced by a factor of at least $\mu$ and the number of heavy edges in $M$ is not changed.
%	\item all paths in $W'$ are not touched by disjunction with $W_t$.

\noindent
{\bf Step-(2).}
Set $I = I_{\leq t-1}$. Note that $|W_{\leq t-1}| = f(Y_{\leq t-1}, I) = f(Y_{\leq t-1}, I_{\leq t-1})$ still holds.

\noindent
{\bf Step-(3).}
Set $l=t$ and repeat the collapse if possible. Remove all unblocked edges in $X_t$ (since $|Y_{t}|$ decreases).
%If it is reached by $W'$, then include it in $I$ since each such unblocked edge $(i,P)$ is addable:
For each removed unblocked edge $e$, include it in $I$ if
%\begin{equation*}
$f(Y_{\leq t-1}, X_{\leq t}\cup I + e)  >  f(Y_{\leq t-1}, X_{\leq t}\cup I)$.
%\end{equation*} 

\subsection{Invariants and Properties}

\begin{fact}[Key Invariant]\label{fact:invariant}
	Since the construction of $L_t$ (until $L_{t-1}$ is collapsed), $f(Y_{\leq t-1},X_{\leq t}\cup I)$ does not decrease and is always no less than $|X_{\leq t}|$.
\end{fact}
\begin{proof}
	We prove by induction on $t\geq 1$. Consider the base case when $t=1$.
	The statement trivially holds when $L_t$ is just constructed and when $|X_t\cup I|$ increases.
	Suppose in some iteration $|X_t\cup I|$ decreases, then it must be because $Y_t$ is collapsed, in which case $f(Y_{\leq t-1}, X_{\leq t}\cup I)$ does not change due to the update rule of step-(3).
	
	Now assume the statement is true for $t$ and consider $t+1$.
	
	When $L_{t+1}$ is built we have $f(Y_{\leq t}, X_{\leq t+1}\cup I)\geq f(Y_{\leq t-1}, X_{\leq t}\cup I\cup X_{t+1}) = f(Y_{\leq t-1}, X_{\leq t}\cup I) + |X_{t+1}| \geq |X_{\leq t+1}|$.
	Since $|X_i|$ does not increase afterwards for all $i\leq t+1$, applying the same argument to $L_{t+1}$ as above yields the fact.
\end{proof}
%Let $d_i = f(Y_{\leq i-1}, X_{\leq i}\cup I)$ right after $L_i$'s construction. Then we have the following invariants at the beginning of each iteration:
%\begin{enumerate}
%	\item $f(Y_{\leq l}, I) = I$ since it increases when $|I|$ increases; remains tight when collapsed.
%	\item $f(Y_{\leq i-1}, X_{\leq i}\cup I)\geq d_i$ for all $i\in[l]$ since $|I|$ increases unless collapsed; $f(Y_{\leq t-1}, X_{\leq t}\cup I)$ does not decrease when $|I|$ and $|X_t|$ are decreased in collapse phase.
%	\item {\color{red} IMPORTANT} $d_i\geq |X_{\leq i}|$ since when $L_{l+1}$ is built we have $d_{l+1}=f(Y_{\leq l}, X_{\leq l+1}\cup I)\geq f(Y_{\leq l-1}, X_{\leq l}\cup I\cup X_{l+1})\geq d_l + |X_l|\geq |X_{\leq l}|$ and $|X_i|$ does not increase afterwards.
%\end{enumerate}

\begin{lemma}[Exponential Growth]\label{lemma:exp-growth}
Let $r=\max\{\lceil \frac{k}{9} \rceil, \lceil\frac{k-10}{3+2\sqrt{2}}\rceil\}$, if $T\leq \opt$, then for all $i\in[l]$ we have $|Y_i|\geq \mu^2|Y_{\leq i-1}|$, which implies $l = O(\frac{1}{\mu^2} \log n)$.
\end{lemma}
\begin{proof}
	First note that we have $(p-r+1)|X_{\leq t}|\leq r|Y_{\leq t}|$ since each edge in $X_{\leq t}$ has at least $p-r+1$ blocked light items.
	Suppose $|Y_t|<\mu^2|Y_{\leq t-1}|$, then we have $f(Y_{\leq t-1},X_{\leq t}\cup I)< (\frac{r}{p-r+1} + 2\mu)|Y_{\leq t-1}|$ since otherwise (the last inequality holds since no collapsible layer):
	\begin{equation*}
	\frac{r}{p-r+1}|Y_{\leq t}|\geq |X_{\leq t}|\geq f(Y_{\leq t-1}, X_{\leq t}\cup I)-|I_{\leq t}|\geq f(Y_{\leq t-1},X_{\leq t}\cup I) - \mu |Y_{\leq t}|,
	\end{equation*}
	which leads to a contradiction (assume $\frac{1}{\mu}\geq \frac{r}{p-r+1} + \mu$):
	\begin{equation*}
	(\frac{r}{p-r+1} + \mu)|Y_t| \geq (\frac{r}{p-r+1} + 2\mu - \frac{r}{p-r+1} - \mu)|Y_{\leq t-1}| \geq \mu|Y_{\leq t-1}|.
	\end{equation*}
	
	Let $\gamma = \frac{r}{p-r+1} + 2\mu$.
	Consider the moment when there is no more addable edge that can be included into $X_{l+1}$ (before adding $Y_{l+1}$).
	Assume $|Y_{l+1}|<\mu^2|Y_{\leq l}|$, then we have $|X_{\leq l+1}|\leq f(Y_{\leq l},X_{\leq l+1}\cup I) < \gamma |Y_{\leq l}|$.
	The current number of light items in the tree is
	\begin{align*}
	|B^\epsilon(Y_{\leq l}\cup X_{\leq l+1}\cup I)| &\leq (r-1)|X_{\leq l+1}| + r|Y_{\leq l+1}| + p|I_{\leq l+1}|\\
	&\leq ((r-1)\gamma + r(1+\mu^2) + \mu(1+\mu^2) p)|Y_{\leq l}|.
	\end{align*}
	
	Consider the residual graph $G'$ of $G(A\cup B^1, E_M)$ with flow $F(Y_{\leq l},X_{\leq l+1}\cup I)$ (obtained by reversing the direction of each path).
	Note that since $f(Y_{\leq l},X_{\leq l+1}\cup I) < \gamma |Y_{\leq l}|$, more than $(1-\gamma)|Y_{\leq l}|$ of $A(Y_{\leq l})$ can reach at least one agent $i\in A$.
	In $G'$, let $T$ be those agents reachable from $A(Y_{\leq l})$.
	For all $i\in T$,  we have $|B^\epsilon_i\backslash B^\epsilon(Y_{\leq l}\cup X_{\leq l+1}\cup I)|\leq p-1$ (otherwise there are addable edges), and for all $j\in B^1_i$, $j$ must be reachable from $A(Y_{\leq l})$ and assigned.
	Hence the total number of heavy items agents in $T$ are interested in is less than $|T|- (1-\gamma)|Y_{\leq l}|$, which means that more than $(1-\gamma)|Y_{\leq l}|$ agents in $T$ are assigned only light items in the optimal solution.
	Note that each such agent is assigned at least $k$ light items and at most $p-1$ of those items are not included in the tree, we have
	\begin{equation*}
	|B^\epsilon(Y_{\leq l}\cup X_{\leq l+1}\cup I)| > (k-p+1)(1-\gamma)|Y_{\leq l}|,
	\end{equation*}
	which implies
	\begin{equation*}
	\frac{r}{p-r+1} = \gamma-2\mu > \frac{k-p-r+1- \mu(2k-(1+\mu^2)p+(2+\mu)r)}{k-p+r}.
	\end{equation*}
	Fix $\mu = 10^{-10}$, the above inequality is always {\bf not} true for all $k\geq 9$, $r=\lceil \frac{k}{9} \rceil$ and $p=3r-1$ by some simple calculation.
	Moreover, as $\epsilon\rightarrow 0$ (which means $k\rightarrow \infty$), we can set $r = \frac{k-10}{3+2\sqrt{2}}$, $p = (2+\sqrt{2})r-1$ such that the above inequality is not true. Hence, we have a contradiction and we claim that we always have $|Y_{l+1}|\geq \mu^2|Y_{\leq l}|$.
\end{proof}

Now we are ready to prove Theorem~\ref{th:poly-9}.

\begin{proofof}{Theorem~\ref{th:poly-9}}
For any $T$ and $k = \lceil \frac{T}{\epsilon} \rceil$, the algorithm tries to compute an $r\epsilon$-allocation, for integer $r$ as large as possible, by enumerating all possible values of $p$ between $r$ and $k$.
For any fixed $r$ and $p$, we try to augment the partial matching $M$ that matches each agent with either a heavy item or $r$ light items.
Hence it suffices to show that the algorithm terminates in polynomial time for augmenting the size of $M$ by one. 
Since each iteration can be done in polynomial time, it suffices to bound the number of iterations by $\poly(n)$.
The approximation ratio will be the maximum of $\frac{k}{r}$, over all $T\leq \opt$.

By Lemma~\ref{lemma:exp-growth} and the definition of collapsible, we know that after each iteration, either (if no collapse) a new layer with $|Y_{l+1}|\geq \mu^2|Y_{\leq l}|$ is constructed, or some $|Y_t|$ is reduced to at most $(1-\mu)|Y_t|$ while $Y_i$ are unchanged, for all $i<t$.
Let $s_i = \lfloor \log_{\frac{1}{1-\mu}}\frac{|Y_i|}{\mu^{2i}} \rfloor$ and $\textbf{s} = (s_1,s_2,\ldots,s_l,\infty)$ be the signature, then we have:
(1) it is lexicographically decreasing across all iterations: if there is no collapse, then some layer is newly constructed and hence $\textbf{s}$ decreases; otherwise let $L_t$ be the last layer that is collapsed and $|Y_t|$ be the size of $Y_t$ before it is collapse: we know that at the end of the iteration $s_i$ is not changed for all $i<t$ while $s_t \leq \lfloor \log_{\frac{1}{1-\mu}}\frac{(1-\mu)|Y_i|}{\mu^{2i}} \rfloor = \lfloor \log_{\frac{1}{1-\mu}}\frac{|Y_i|}{\mu^{2i}} \rfloor - 1$ is decreased by at least one, which also means $\textbf{s}$ decreases;
(2) its coordinates are not decreasing: for all $i\in[l-1]$ we have $s_{i+1} = \lfloor \log_{\frac{1}{1-\mu}}\frac{|Y_{i+1}|}{\mu^{2i+2}} \rfloor \geq \lfloor \log_{\frac{1}{1-\mu}}\frac{|Y_{\leq i}|}{\mu^{2i}} \rfloor \geq s_i$.
Since we have $l = O(\log n)$ and $s_i = O(\log n)$ for all $i\in[l]$, the total number of iterations (signatures) is at most $2^{O(\log n)} = \poly(n)$.

\paragraph{Approximation Ratio.}
When $k\leq 9$, then a trivial $9$-approximation can be achieved by a $\epsilon$-allocation (maximum matching).
By the proof of Lemma~\ref{lemma:exp-growth}, the approximation ratio $\frac{k}{r}$ is always at most $9$ and tends to $3+2\sqrt{2}\approx 5.83$ as $\epsilon\rightarrow 0$.
\end{proofof}

\section{Hardness of $(1,\epsilon)$-Restricted Allocation Problem} \label{sec:hardness}

We show that for any $\epsilon\leq 0.5$, the $(1,\epsilon)$-restricted allocation problem cannot be approximated within any ratio smaller than $2$.

\begin{definition}[$3$-dimensional matching]
	Given a $3$-regular hypergraph $H(X\cup Y\cup Z, E)$ where $|X|=|Y|=|Z|$ and $E\subseteq X\times Y\times Z$, the $3$-dimensional matching problem aims at finding a perfect matching $M\subseteq E$ that matches all nodes.
\end{definition}

\begin{proofof}{Theorem~\ref{th:hardness}}
Deciding the existence of a perfect matching in the $3$-dimensional matching problem is known to be \NP-hard.
Given an instance of the $3$-dimensional matching problem $H(X\cup Y\cup Z, E)$, for any fix $\epsilon\leq 0.5$, we show that there exists an instance $(A,B,w)$ of the $(1,\epsilon)$-restricted allocation problem for which $\opt = 2\epsilon$  if $H$ has a perfect matching; otherwise $\opt \leq \epsilon$.
Hence no polynomial time algorithm can approximate the $(1,\epsilon)$-restricted allocation problem within any ratio smaller than $2$, unless \P=\NP.

Let $d(z)$ be the number of hyperedges adjacent to node $z\in Z$. Define $\hat{Z} = \{ z^{(1)}, z^{(2)},\ldots,$ $z^{(d(z)-1)}:z\in Z \}$ to be the set containing $d(z)-1$ copies of each $z\in Z$. Let $A = E$, $B = X\cup Y\cup \hat{Z}$, $w_j = \epsilon$ for all $j\in X\cup Y$ and $w_j = 1$ for all $j\in \hat{Z}$.
For all $e = (x,y,z)\in A$, let $B_e = \{ x,y,z^{(1)}, z^{(2)},\ldots,z^{(d(z)-1)} \}$.

Since there are $|E|$ agents and $\sum_{z\in Z}(d(z)-1) = |E|-|Z|$ heavy items, at least $|Z|$ agents receive only light items. Since there are $2|Z|$ light items, we have $\opt\leq 2\epsilon$.

{\bf YES case.} If $H$ has a perfect matching $M$, then for each $e = (x,y,z)\in M$, we can assign $x,y\in B_e$ to $e\in A$. For the remaining $|E|-|Z|$ agents, we can assign to each agent one heavy item in $\hat{Z}$ (since exactly $d(z)-1$ edges adjacent to each $z\in Z$ are not assigned). Hence, we have $\opt = 2\epsilon$.

{\bf NO case.} If there is no perfect matching then we show that $\opt<2\epsilon$, which means $\opt \leq \epsilon$.
Assume the contrary that $\opt = 2\epsilon$. Then every $e= (x,y,z)\in A$ must receive either a single heavy item or two light items, which must be $x$ and $y$. Since every $z\in Z$ has only $d(z)-1$ copies, at least one of the adjacent edges of $z$ must receive no heavy item, which means the edges receiving light items having disjoint $Z$ nodes. Hence, those $|Z|$ edges receiving light items actually form a perfect matching and it is a contradiction.
\end{proofof}

While the above analysis implies that the integrality gap of $\clp$ is at least $2$ when $\P \neq \NP$, our following example shows that the integrality gap is, unconditionally, at least $2$.

\paragraph{Lower Bound for the Integrality Gap.}
For the $(1,\epsilon)$-restricted allocation problem instance in
Figure~\ref{fig:IG2} with $4$ agents (circles) and $6$ items (squares), $T^*= 2\epsilon$ while $\opt = \epsilon$ (since at least one light item will become useless after assigning all heavy items), which implies that the integrality gap is at least $2$.

\begin{figure}[htb]
	\centering
	\includegraphics*[width=0.3\textwidth]{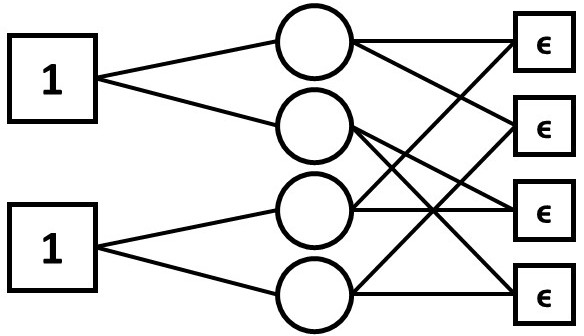}
	\caption{$(1,\epsilon)$-restricted allocation problem instance with integrality gap $2$}
	\label{fig:IG2}
\end{figure}

{
\bibliography{hypergraph}
\bibliographystyle{plain}
}

\end{document}